\documentclass[reqno,a4paper,11pt]{amsart}
\usepackage{amsmath,amstext,amsfonts,amsbsy,eucal,amssymb}
\usepackage[latin1]{inputenc}

\numberwithin{equation}{section}

\newtheorem{theorem}{Theorem}[section]
\newtheorem{lemma}[theorem]{Lemma}

\newtheorem{proposition}[theorem]{Proposition}
\newtheorem{corollary}[theorem]{Corollary}

\theoremstyle{definition}
\newtheorem{definition}[theorem]{Definition}
\newtheorem{example}[theorem]{Example}
\newtheorem{remark}[theorem]{Remark}

\newcommand{\C}{\mathbf{C}}

\newcommand{\Z}{\mathbf{Z}}

\newcommand{\Mg}{M_g}
\newcommand{\Hg}{H_g}
\newcommand{\Phig}{\Phi_g}

\newcommand{\Gr}{\mathop{\rm Gr}\nolimits}

\newcommand{\Sym}{\mathop{\rm Sym}\nolimits}
\newcommand{\Div}{\mathop{\rm Div}\nolimits}
\newcommand{\divi}{\mathop{\rm div}\nolimits}
\newcommand{\aj}{\mathop{\rm aj}\nolimits}

\newcommand{\pt}{\operatorname{\partial}}
\newcommand{\Image}{\operatorname{Im}}
\newcommand{\diff}{{\rm d }}
\newcommand{\X}{\mathcal X}

\newcommand{\cO}{\mathcal O}

\newcommand{\PB}{\left\{\cdot\,,\cdot\right\}}

\newcommand{\Pb}[1]{\left\{\cdot\,,#1\right\}}

\begin{document}

\parskip 4pt
\baselineskip 16pt

%%%%%%%%%%%%%%%%%%%%%%%%%%%%%%%%%
%%%%%%%%%%%%%%%%%%%%%%%%%%%%%%%%%

\title[Singular fiber of the Mumford system]
{Singular fiber of the Mumford system 
\\ and rational solutions to the KdV hierarchy}

\author[Rei Inoue]{Rei Inoue${}^1$}
\address{Rei Inoue, 
Faculty of Pharmaceutical Sciences,
Suzuka University of Medical Science, 3500-3 Minami-tamagaki, Suzuka, 
Mie 513-8670, Japan}
\email{reiiy@suzuka-u.ac.jp}
\thanks{${}^1$Partially supported by 
Grand-in-Aid for Young Scientists (B) (19740231)}
\author[Pol Vanhaecke]{Pol Vanhaecke${}^2$}
\address{Pol Vanhaecke, Laboratoire de Math\'ematiques et Applications, 
UMR 6086 du CNRS, Universit\'e de Poitiers, 
Boulevard Marie et Pierre CURIE, BP 30179, 86962 Futuroscope Chasseneuil
Cedex, France}
\email{pol.vanhaecke@math.univ-poitiers.fr}
\thanks{${}^2$Partially supported by 
a European Science Foundation grant (MISGAM), 
a Marie Curie grant (ENIGMA) and an ANR grant (GIMP)}

\author[Takao Yamazaki]{Takao Yamazaki${}^3$} 
\address{Takao Yamazaki, Mathematical Institute, Tohoku University,
  Aoba, Sendai 980-8578, Japan}
\email{ytakao@math.tohoku.ac.jp}
\thanks{${}^3$Partially supported by 
Grand-in-Aid for Young Scientists (B) (19740002)}

%\date{\today}
\subjclass[2000]{53D17, 37J35, 14H70, 14H40}

\keywords{Integrable systems; Generalized Jacobians; Rational solutions}

% ----------------------------------------------------------------

\begin{abstract}
We study the singular iso-level manifold $M_g(0)$
of the genus $g$ Mumford system
associated to the spectral curve $y^2=x^{2g+1}$.  
We show that $M_g(0)$ is stratified by $g+1$ open subvarieties of additive
algebraic groups of dimension $0,1,\dots,g$ 
and we give an explicit description of $M_g(0)$ in terms of the
compactification of the generalized Jacobian.
As a consequence, we obtain an effective algorithm 
to compute rational solutions to the genus $g$ Mumford system,
which is closely related to rational solutions of the KdV hierarchy.
\end{abstract}

\maketitle

%%%%%%%%%%%%%%%%%%%%%%
\section{Introduction}
%%%%%%%%%%%%%%%%%%%%%%

The notion of algebraic integrability has been introduced by Adler and van Moerbeke in order to provide a natural
context in which basically all classical examples of integrable systems naturally fit (after complexification) and
they have developed techniques for studying the geometry and the explicit integration of these systems
\cite{adler_adv, adler_adv2, avm}. The main feature of an algebraic completely integrable system (a.c.i.\ system)
is that the \emph{generic} fiber of its complex momentum map (the map which is defined by the Poisson commuting
integrals) is an affine part of an Abelian variety (compact complex algebraic torus); in addition, the
corresponding Hamiltonian vector fields are demanded to define translation 
invariant vector fields on these
tori. One important consequence is that the integration of the equations of motion, starting from a \emph{generic}
point, can be done in terms of theta functions, such as the classical Riemann theta function. A widely known
example of an a.c.i.\ system is the Euler top, which Euler integrated in terms of elliptic functions.

Particular \emph{special (non-generic)} fibers of a the moment map of an a.c.i.\ system are in general not affine
parts of an Abelian variety. According to a conjecture, stated in \cite[p.\ 155]{avm}, such a fiber is made up by
affine parts of one or several algebraic groups, defined by the flows of the integrable vector fields. The
solutions starting from a point on such a fiber are then expressed in terms of a degeneration of the theta
function, such as exponential or rational functions. When the generic fiber of the a.c.i.\ system is the Jacobian
of a Riemann surface, so that the solution is expressed in terms of its Riemann theta function, one is tempted to
relate the algebraic groups that make up a special fiber to a generalized Jacobian, i.e., the Jacobian of a
singular algebraic curve.  Then the function theory of these Jacobians provides the algebraic functions in which
the corresponding special solution can be expressed.  In the case which we will study in this paper, the zero-fiber
of the genus $g$ Mumford system, the singular curve is of the form $y^2=x^{2n+1}$ (where $n\leqslant g$) and the
entire zero-fiber admits, according to a result by Beauville \cite{Beauville90}, a natural description as an affine
part of the compactification of the generalized Jacobian of the curve $y^2=x^{2g+1}$. We will show that the
corresponding solutions of the Mumford system are rational functions of all time variables and we will give
explicit formulas for these solutions. See \cite{audin,fedorov,gavrilov} for other works on integrable systems
involving generalized Jacobians.

Recall \cite{Mumford-Book,Van1638} that for a fixed positive integer $g$, the phase space $\Mg$ of the Mumford
system is given by
\begin{align}\label{Mg}
  \Mg =
  \left\{ \ell(x) = 
  \begin{pmatrix}
    v(x) & w(x) \\
    u(x) & -v(x)
  \end{pmatrix}
  ~\Bigl|~
  \begin{array}{l}
    u(x) = x^g + u_{g-1} x^{g-1} + \cdots + u_0, \\
    v(x) = v_{g-1} x^{g-1} + \cdots + v_0, \\
    w(x) = x^{g+1} + w_g x^g + \cdots + w_0
  \end{array}
  \right\} 
  (\cong \C^{3g+1}),
\end{align}
equipped with a Poisson structure $\PB$. We have the momentum map 
\begin{equation*}
  \Phig:\Mg\to\Hg \ : \
  \ell(x) \mapsto -\det (\ell(x)),
\end{equation*}
where $\Hg \cong \C^{2g+1}$ is given by 
\begin{equation}\label{Hg}
  \Hg=\{h(x) =x^{2g+1} + h_{2g} x^{2g} 
      + h_{2g-1} x^{2g-1}+ \cdots + h_0 ~|~ h_0,\dots,h_{2g}\in\C \}.
\end{equation}
Out of the $2g+1$ independent functions $h_0, \dots, h_{2g+1}$ on $\Mg$, $g+1$ functions $h_g, \dots, h_{2g}$ are
Casimirs, and the $g$ other functions $h_0, \dots, h_{g-1}$ define commuting Hamiltonian vector fields $\X_1,
\cdots, \X_g$.  This implies, since the generic rank of $\PB$ is $2g$ on $M_g$, that the system $(\Mg, \PB, \Phig)$
is a Liouville integrable system. For $h(x)\in\Hg$, let $C_g(h)$ denote the integral projective (possibly singular)
hyperelliptic curve of (arithmetic) genus $g$, given by the completion of the affine curve $y^2 = h(x)$ with one
smooth point at infinity.  The main feature of the Mumford system is that, when $C_g(h)$ is non-singular, there is
an isomorphism between the level set $M_g(h) := \Phig^{-1}(h)$ and the complement of the theta divisor in the
Jacobian variety $J_g(h)$ of $C_g(h)$, which transforms the Hamiltonian vector fields $\X_1, \cdots, \X_g$ into the
translation invariant vector fields on $J_g(h)$.
%Since $C_g(h)$ is, for \emph{generic} $h\in\Hg$, non-singular
This shows that the Mumford system is a.c.i. 
For singular curves, according to Beauville \cite{Beauville90}, the same result holds true, upon
replacing the Jacobian by the compactified generalized Jacobian (and the theta divisor by its completion in the
latter).

In this paper we give a precise and explicit description of the zero-fiber of the Mumford system, which is the
fiber of $\Phig$ over the very special point $h(x) = x^{2g+1}$ in $\Hg$, for which the spectral curve $C_g :=
C_g(x^{2g+1})$ becomes a singular curve given by $y^2=x^{2g+1}$.
% with generalized Jacobian $J_g$. 
Our results can be summarized as follows.  
(See Theorems \ref{cor_beauville}, \ref{main} and Proposition \ref{prop:uvw-tau} for (1)-(3).)
\begin{enumerate}
\item
The level set $M_g(x^{2g+1})$ is stratified by $g+1$ smooth affine varieties, which are invariant for the flows of the
vector fields $\X_1, \dots, \X_g$; they are of dimension $k=0, 1, \dots, g$.
\item
Let $k \in \{0, 1, \dots, g\}$.  There is an isomorphism between the (unique) $k$-dimensional invariant manifold in
$M_g(x^{2g+1})$ and the complement of the `theta divisor' $\Theta_k$ in the generalized Jacobian $J_k$ of $C_k$,
which linearizes the vector fields $\X_1, \dots, \X_k$.  (The vector fields $\X_{k+1}, \dots, \X_g$ vanish.)  On
the other hand, we construct explicitly an isomorphism between $J_k$ and the additive group $\C^k$, by which
$\Theta_k$ is transformed to the zero locus of an (explicitly constructed) polynomial function $\tau_k$ on $\C^k$.
Combined, for $k=g$, this yields a rational solution to the Mumford system in terms of $\tau_g$ and its
derivatives.
\item
The entire level set $M_g(x^{2g+1})$ is isomorphic to the complement of the `completed theta divisor'
$\bar{\Theta}_g$ in the compactification $\bar{J}_g$ of $J_g$.  The vector fields $\X_1, \dots, \X_g$ are
transformed to the ones induced by the natural action of $J_g$ on $\bar{J}_g$ via this isomorphism.
\end{enumerate}

The rational solutions, obtained in (2), turn out to be exactly same as the rational solutions to the Korteweg-de
Vries (KdV) hierarchy constructed in \cite{Ablo,Adler-Moser78,AMM77,Krichever79,Nimmo,sw}.  This is not surprising,
since Mumford's original motivation for constructing the Mumford system is the fact that \emph{every} solution to
the Mumford system yields a solution to the KdV hierarchy \cite[p.\ 3.203]{Mumford-Book}. We therefore recover the
rational solutions of the KdV hierarchy by using an adapted version of the Abel-Jacobi map within the
finite-dimensional framework of the Mumford system.

\noindent
{\it Outline of the paper.}
In \S 2, we briefly review the basic facts about the Mumford system.  \S 3 is devoted to a detailed analysis of the
generalized Jacobian $J_g$ of $C_g$ and its compactification~$\bar{J}_g$.  We then apply in \S 4 the results of \S
3 to the Mumford system.  In \S 5, we give an algorithm to produce rational solutions for the Mumford system.  In
\S 6, we study the relation to the KdV hierarchy.

\noindent
{\it Acknowledgement.}  Part of this work is done while the first and third authors stay in the Universit\'e de
Poitiers.  They are grateful to the hospitality of the members there. We also wish to thank the anonymous referee
for his suggestions which allowed us to better relate our results to the extensive KdV literature.

%%%%%%%%%%%%%%%%%%%%%%%%%%%%%%%%%%%%%%%%%%%%%%%
\section{The Mumford system}\label{sec:mumford}
%%%%%%%%%%%%%%%%%%%%%%%%%%%%%%%%%%%%%%%%%%%%%%%

In this section, we recall the basic facts about the Mumford system (\cite{Mumford-Book},
\cite[Ch. VI.4]{Van1638}).  Throughout the section, $g$ is a fixed positive integer.

%%%%%%%%%%%%%%%%%%%%%%%%%%%%%%%%%%%%%%%%%%%%%%%%%%%%
\subsection{Hamiltonian structure and integrability}
%%%%%%%%%%%%%%%%%%%%%%%%%%%%%%%%%%%%%%%%%%%%%%%%%%%%
The phase space $\Mg$ defined in \eqref{Mg} of the Mumford system
is equipped with the Poisson structure defined by
(see \cite[Ch. VI (4.4)]{Van1638}) 
\begin{align*}
    &\{ u(x), u(z) \} = \{ v(x), v(z) \} = 0,\\
    \displaybreak[0]
    &\{ u(x), v(z) \} = \frac{u(x) - u(z)}{x-z},\\
    \displaybreak[0]
    &\{ u(x), w(z) \} = -2 \frac{v(x) - v(z)}{x-z},\\
    \displaybreak[0]
    &\{ v(x), w(z) \} = \frac{w(x) - w(z)}{x-z} - u(x) ,\\
    \displaybreak[0]
    &\{ w(x), w(z) \} = 2 \bigl(v(x) - v(z)\bigr).
\end{align*}
The natural coordinates $h_0,\dots,h_{2g}$ on $\Hg$ \eqref{Hg} 
can be regarded as polynomial functions on $\Mg$.
These functions are pairwise in involution with respect to the above 
Poisson structure\footnote{Actually, they are in involution 
with respect to a whole family of compatible Poisson
structures, see \cite[Ch. VI (4.4)]{Van1638}.},
where $h_{g}, \cdots, h_{2g}$ are the Casimirs,
and $h_{0},\cdots, h_{g-1}$ generate the Hamiltonian vector fields 
$\X_1,\cdots, \X_g$ on $\Mg$ by $\X_{i}:=\Pb{h_{g-i}}$.
Introducing $D(z) := \sum_{i=0}^{g-1} z^{i}\,\X_{g-i}$, 
these vector fields can be
simultaneously written as follows (see \cite[Th. 3.1]{Mumford-Book}):
\begin{align}\label{time-evol}
  \begin{split}
    &D(z) u(x) = 2 \, \frac{u(x) v(z) - v(x) u(z)}{x-z},\\
    &D(z) v(x) = \frac{w(x) u(z) - u(x) w(z)}{x-z} - u(x) u(z),\\
    &D(z) w(x) = 2 \Bigl(\frac{v(x) w(z) - w(x) v(z)}{x-z}+ v(x) u(z) \Bigr).
  \end{split}
\end{align}
Since $\Phig$ is submersive and since the above $g$ Hamiltonian vector fields are independent at a generic point of
$\Mg$, a simple count shows that the triplet $(\Mg,\PB,\Phig)$ is a (complex) Liouville integrable system.
%
%%%%%%%%%%%%%%%%%%%%%%%%%%%%%%%%%%%%
\subsection{Algebraic integrability}
%%%%%%%%%%%%%%%%%%%%%%%%%%%%%%%%%%%%
%
It was shown by Mumford that $(\Mg,\PB,\Phig)$ is actually an a.c.i. system, which means that, in addition to
Liouville integrability, the generic fiber of the momentum map $\Phig$ is an affine part of an Abelian variety
(complex algebraic torus), and that the above Hamiltonian vector fields are constant (translation invariant) on
these tori.  We sketch the proof, which Mumford attributes to Jacobi.  To a polynomial $h(x)\in\Hg$, one naturally
associates two geometrical objects:
\begin{enumerate}
  \item[$\bullet$] 
The {\it spectral curve} $C_g(h)$ is defined to be 
a completion of the affine curve in $\C^2$ given by $y^2 = h(x)$
by adding one smooth point $\infty$.
This is an integral projective (possibly singular) hyperelliptic curve
of (arithmetic) genus $g$.
\item[$\bullet$] 
The {\it level set} $M_g(h)$ is defined to be
the fiber of $\Phig$ over $h(x)$.
\end{enumerate}

\begin{theorem}[Mumford]\label{thm:mumford}
Suppose that $h(x) \in H_g$ has no multiple roots, so that $C_g(h)$ is an irreducible projective smooth
hyperelliptic curve of genus $g$.  Let $J_g(h)$ and $\Theta_g(h)$ be the Jacobian variety and the theta divisor of
$C_g(h)$.  Then there is an isomorphism $M_g(h) \cong J_g(h) \setminus \Theta_g(h)$ by which the vector fields
$\X_1, \dots, \X_{g}$ are transformed into independent translation invariant vector fields on $J_g(h)$.
\end{theorem}

\noindent
{\it Outline of the proof.}  One first proves that there is an isomorphism between $M_g(h)$ and an open dense
subset
\begin{equation*}
\mathcal{S}:= 
  \{ \sum_{i=1}^g [P_i] \in \Sym^g(C_g(h) \setminus \{ \infty \})
     ~|~ i\neq j\Rightarrow P_i\neq\imath (P_j) \}
\end{equation*}%
of $\Sym^g(C_g(h))$,
where $\imath: C_g(h) \to C_g(h)$ is the hyperelliptic involution.
This isomorphism is given by
\begin{equation}\label{eq:mumford_map_def}
  \ell(x)=\begin{pmatrix}
    v(x) & w(x) \\
    u(x) & -v(x)
  \end{pmatrix} 
  \mapsto \sum_{\hbox{roots $x_i$ of $u(x)$}} [ (x_i,v(x_i)) ]
\end{equation}%
when $u(x)$ has no multiple roots, which naturally extends to the whole of $M_g(h)$ by the interpolation formula.
The next step is to show that the Abel-Jacobi map induces an isomorphism between $\mathcal{S}$ and $J_g(h)
\setminus \Theta_g(h)$. Combined with the first step, this yields the isomorphism between $M_g(h)$ and
$J_g(h)\setminus \Theta_g(h)$.

As for the translation invariance of the vector fields $\X_1, \dots, \X_{g}$ on $J_g(h)$, it suffices to prove that
they are translation invariant in the neighborhood of a generic point, because they are holomorphic on $M_g(h)$. 
%, it suffices to show the equality on an open dense subset of
%$\mathcal{S}$, because both vector fields in question are holomorphic.  
We use the above isomorphism to write these Hamiltonian vector fields down in terms of the variables $x_i$, which
yield local coordinates in the neighborhood of a generic point of $\mathcal S$.  We calculate $D(z)u(x_i)$ in two
different ways:
\begin{align*}
  & D(z) u(x) \big|_{x=x_i} = 2 \, \frac{v(x_i) u(z)}{z-x_i} 
  = 2 y_i \prod_{k \neq i} (z-x_k),
\\  
  &D(z) u(x) \big|_{x=x_i} = -\prod_{k \neq i} (x_i-x_k) D(z) x_i.
\end{align*}
Thus 
$$
  \frac{D(z) x_i}{y_i} = -2 \prod_{k \neq i} \frac{z-x_k}{x_i-x_k}.
$$
By using the interpolation formula, we obtain
$$
  \sum_{i=1}^g x_i^{j-1}\frac{D(z) x_i}{y_i}
  =-2\sum_{i=1}^g x_i^{j-1}\prod_{k\neq i}\frac{z-x_k}{x_i-x_k}=-2 z^{j-1}
$$
for $j=1,\cdots,g.$ 
It follows that in terms of the local coordinates $x_i$, 
the vector fields $\X_i$ are
expressed by
\begin{align}\label{vf-xy}
  \begin{pmatrix}
    1 & 1 & \cdots & 1 \\ 
    x_1 & x_2 & \cdots & x_g \\
    \vdots & & & \vdots \\
    x_1^{g-1} & x_2^{g-1} & \cdots & x_g^{g-1} 
  \end{pmatrix}
  \begin{pmatrix}
    \frac{\X_g x_1}{y_1} & 
        \frac{\X_{g-1} x_1}{y_1} & \cdots & \frac{\X_{1} x_1}{y_1} \\
    \frac{\X_g x_2}{y_2} & \cdots & & \frac{\X_1 x_2}{y_2} \\
    \vdots & & & \vdots \\
    \frac{\X_g x_g}{y_g} & \cdots & & \frac{\X_1 x_g}{y_g} 
  \end{pmatrix}
  =
  -2 \, \mathbb{I}_g.
\end{align}
The $g$ differential forms $\{ \sum_{i=1}^g x_i^j\diff x_i/y_i \}_{j=0,\dots,g-1}$ on $\mathcal{S}$ are seen to be
the dual basis to $\{ \X_i \}_{i=1, \dots, g}$ (up to a scalar) by \eqref{vf-xy}.  Since $\{ \sum_{i=1}^g
x_i^j\diff x_i/y_i \}_{j=0,\dots,g-1}$ constitute under the Abel-Jacobi map a basis for the space of holomorphic
one-forms on $J_g(h)$, it follows that $\X_1, \dots, \X_g$ extends to holomorphic (hence translation invariant)
vector fields on $J_g(h)$.
\qed

%%%%%%%%%%%%%%%%%%%%%%%%%%%
\subsection{Singular fiber}
%%%%%%%%%%%%%%%%%%%%%%%%%%%
We consider what happens in Theorem \ref{thm:mumford} 
when $C_g(h)$ is singular.
For a coherent sheaf $\mathcal{F}$ on $C_g(h)$
and $k \in \Z$,
we write $\mathcal{F}(k)$ for 
$\mathcal{F} \otimes \cO_{C_g(h)}(k [\infty])$.
For any $h(x) \in \Hg$, we define $J_g(h)$ and $\bar{J}_g(h)$ 
respectively
to be the generalized Jacobian variety of $C_g(h)$
(which parametrizes invertible sheaves on $C_g(h)$ of degree zero)
and its compactification
(which parametrizes torsion free $\cO_{C_g(h)}$-modules $\mathcal{L}$ 
of rank one such that 
$h^0(C_g(h), \mathcal{L})-h^1(C_g(h), \mathcal{L})=1-g$)
(see \cite{dsouza}).
We have $J_g(h) = \bar{J}_g(h)$ if $h(x)$ has no multiple root.
We have a natural inclusion $J_g(h) \subset \bar{J}_g(h)$
(see \cite{rego}).
We also define
\begin{align}
 \Theta_g(h) &:= \{ L \in J_g(h) ~|~h^0(C_g(h), L(g-1)) \not= 0 \},\label{def:theta_1}\\
 \bar{\Theta}_g(h) &:= \{ \mathcal{L} \in \bar{J}_g(h)~|~ h^0(C_g(h), \mathcal{L}(g-1)) \not= 0 \}.\label{def:theta_2}
\end{align}
Note that we have $h^0(C_g(h), \mathcal{L}(g-1)) = h^1(C_g(h), \mathcal{L}(g-1))$ for any $\mathcal{L} \in
\bar{J}_g(h)$.\footnote{When $\mathcal{L} \in J_g(h)$, this is a consequence of the Riemann-Roch theorem (cf.\
\cite{Serre}).  For a general $\mathcal{L} \in \bar{J}_g(h)$, the proof can be reduced to the previous case,
because one can find a partial normalization $f: C' \to C_g(h)$ and an invertible sheaf $L$ on $C'$ of degree zero
such that $f_*(L)=\mathcal{L}$ (cf.\ \cite[p.\ 101]{Beauville99} ).}  We set
\[ M_g(h)_{reg} := 
\{ l(x) \in M_g(h) 
~|~ l(a) ~\text{is regular for all}~a \in \mathbb{P}^1\}.
\]
(Recall that $A \in M_2(\C)$ is {\it regular} iff all eigenspaces of $A$ are one-dimensional.  Note that the matrix
$l(\infty) = \begin{pmatrix} 0 & 1 \\ 0 & 0\end{pmatrix}$ of leading coefficients is regular.)  When $h(x)$ has no
multiple root, we have $M_g(h) = M_g(h)_{reg}.$ Here we state a special case of a result of Beauville
\cite{Beauville90}.

\begin{theorem}[Beauville]\label{thm:beauville}
For any $h\in H_g$, there exist isomorphisms
  \[
    M_g(h)\cong\bar{J}_g(h)\setminus \bar{\Theta}_g(h)\qquad\text{and}\qquad M_g(h)_{reg}\cong J_g(h)\setminus\Theta_g(h),
  \]
   where the latter is a restriction of the former.
\end{theorem}
\noindent
{\it Outline of the proof.}  Let $f: C_g(h) \to \mathbb{P}^1$ be the map given by $(x,y) \mapsto x$.  We take
$\mathcal{L} \in \bar{J}_g(h) \setminus \bar{\Theta}_g(h)$.  We see that the condition $\mathcal{L}
\not\in\bar{\Theta}_g(h)$ implies that there exists an isomorphism $E:=f_*(\mathcal{L}(g-1)) \cong
\mathcal{O}_{\mathbb{P}^1}(-1)^{\oplus 2}$ of $\mathcal{O}_{\mathbb{P}^1}$-modules (and vice versa), which is
unique up to the conjugation by an element of $GL_2(\C)$.  Once we fix this isomorphism, the map $E \to E(g+1)$
defined by the multiplication by $y \in \mathcal{O}_{C_g(h)}$ is represented by a matrix $\tilde{l}(x) \in
M_2(\C[x])$ such that all the entries of $\tilde{l}(x)$ are of degree $\leqslant g+1$.  We also have $-\det
\tilde{l}(x) = h(x)$ by the Cayley-Hamilton formula.  In the $GL_2(\C)$-conjugate class of $\tilde{l}(x)$, there
exists a unique $l(x)$ which belongs to $M_g(h)$ (cf. \cite[(1.5)]{Beauville90}).  It follows that the
correspondence $\mathcal{L} \mapsto l(x)$ defines a bijection $\bar{J}_g(h) \setminus \bar{\Theta}_g(h) \cong
M_g(h)$.  In order to see this is an isomorphism, we simply notice that the same argument works after any base
change.  It is shown in \cite[(1.11-13)]{Beauville90} that the restriction of this isomorphism defines
$M_g(h)_{reg} \cong J_g(h) \setminus \Theta_g(h)$.
\qed

\begin{remark}\label{MandB}
We briefly explain that the two isomorphisms constructed by Mumford and Beauville coincide when $h(x) \in H_g$ has
no multiple root.  We take $l(x) =
\begin{pmatrix} 
  v(x) & w(x) \\ u(x) & -v(x) 
\end{pmatrix} 
\in M_g(h)$.
Mumford associates to $l(x)$ the invertible sheaf $L= \mathcal{O}_{C_g(h)}(D - g [\infty])$ where $D = \sum_{i=1}^g
[(x_i, v(x_i))]$ with $u(x)=\prod_{i=1}^g(x-x_i)$.  We set $E:=f_*(L(g-1)) =
f_*(\mathcal{O}_{C_g(h)}(D-[\infty]))$.  Then we can choose an isomorphism $E(1) \cong
\mathcal{O}_{\mathbb{P}^1}^{\oplus 2}$ in such a way that on the global section $u(x)$ and $y-v(x)$ are mapped to
the standard basis of $\mathcal{O}_{\mathbb{P}^1}^{\oplus 2}$.  (Note that $\{ u(x), y-v(x) \}$ is a basis of
$H^0(C_g(h), L(g))$.)  Then the multiplication by $y$ is represented by $l(x)$, since $(u(x), ~ y-v(x)) y = (u(x),
~y-v(x)) l(x)$ follows from the relation $y^2=h(x)=u(x)w(x)+v(x)^2$.
\end{remark}
Beauville also showed that the Hamiltonian vector fields are transformed by this isomorphism to the vector fields
generated by the group action of $J_g(h)$, but the proof works only when $C_g(h)$ is non-singular.  It should be
possible to modify his argument to deal with singular cases, but we avoid it.  Instead, we limit ourselves to
consider a very singular rational curve obtained by taking $h(x)=x^{2g+1}$, so that the curve is given by
$y^2=x^{2g+1}$.  For this curve we will make the above isomorphism explicit, which entails in particular an
explicit description of $M_g(h)_{reg}$ as a subset of $M_g(h),$ a description of the Jacobian variety as the
additive group $\C^g$, and a description of the theta divisor as a subvariety of $\C^g$.  The latter two
descriptions will be given in the following section.  We will then discuss the Hamiltonian vector fields in \S 4.

%%%%%%%%%%%%%%%%%%%%%%%%%%%%%%%%%%%%%%%%%%%%%%%%%%%%%%%
\section{Generalized Jacobian and its compactification}
%%%%%%%%%%%%%%%%%%%%%%%%%%%%%%%%%%%%%%%%%%%%%%%%%%%%%%%

For a positive integer $g$,
we define $C_g$ to be the (complete, singular) hyperelliptic curve
defined by the equation $y^2 = x^{2g+1}$.
In this section, we study in detail the structure of 
the generalized Jacobian of $C_g$ and its compactification.

%%%%%%%%%%%%%%%%%%%%%%%%%%%%%%%%%
\subsection{Generalized Jacobian}
%%%%%%%%%%%%%%%%%%%%%%%%%%%%%%%%%

Let $J_g$ be 
the generalized Jacobian variety of $C_g$, 
which parametrizes isomorphism classes of invertible sheaves 
on $C_g$ of degree zero (cf. \cite{Serre}).

The normalization of $C_g$ is given by
$\pi_g: \mathbb{P}^1 \to C_g;~ \pi_g(t) = (t^2, t^{2g+1})$.
Let $O$ and $\infty$ be the points on $\mathbb{P}^1$
whose coordinates are $t=0$ and $\infty$ respectively.
The images of $O$ and $\infty$ by $\pi_g$ are, 
by abuse of notation, 
written by the same letter $O$ and $\infty$.
Note that $O$ is the unique singular point on $C_g$.
We write $R_g$ for the local ring 
$\cO_{C_g, O}$ of $C_g$ at $O$,
which we regard as a subring of 
$S = \cO_{\mathbb{P}^1, O} = \C[t]_{(t)}$ (via $\pi_g$).
The completions of $S$ and $R_g$ are 
identified with $\C[[t]]$ and $\C[[t^2, t^{2g+1}]]$ respectively.
The following isomorphisms play an important role throughout this paper
(see, for example, \cite{Serre}):
\begin{equation}\label{jacobian}
 \C^g
 \cong \C[[t]]^*/\C[[t^2, t^{2g+1}]]^*
 \cong S^*/R_g^*
 \cong \Div^0(C_g \setminus \{ O\})/ \divi(R_g^*)
 \cong J_g.
\end{equation}
Here the first map is given by
$$ \vec{a} = (a_1, \dots, a_g) \mapsto 
f(t; \vec{a}) := \exp(\sum_{i=1}^g a_i t^{2i-1}) 
\mod \C[[t^2, t^{2g+1}]]^*.
$$
The second map is induced by 
the "inclusion to their completion"
$S \subset \C[[t]]$ and $R_g \subset \C[[t^2, t^{2g+1}]]$.
The third map associates to the class of $f \in S^*$
its divisor class $\divi(f)$.
The fourth map is defined by $D \mapsto \cO(-D)$,
where for $D \in \Div(C_g \setminus \{ O \})$
we write the corresponding invertible sheaf by $\cO(D)$.
We often identify all the five groups 
appearing in \eqref{jacobian} altogether.

It is convenient to introduce the following notations:

\begin{definition}\label{def:chi}
\begin{enumerate}
\item
We define polynomials $\chi_n \in \C[a_1, a_2, \dots]$
for $n \in \Z_{\geqslant 0}$ by
\begin{align}\label{chi}
   \exp(\sum_{i=1}^{\infty} a_i t^{2i-1}) 
   = \sum_{n=0}^{\infty} \chi_n t^n
\qquad \text{in} ~\C[[t]].
\end{align}
For example,
we have $\chi_0=1, \chi_1=a_1, \chi_2=\frac{a_1^2}{2}, 
\chi_3=\frac{a_1^3}{6} + a_2.$
One sees that $\chi_n$ is a polynomial in
the variables $a_1, \dots, a_{[\frac{n+1}{2}]}$.
We set $\chi_{n}=0$ for $n \in \Z_{<0}$.
\item
We set $f_{g}(t; \vec{a}) := \sum_{n=0}^{2g-1} \chi_n(\vec{a}) t^n$
for $\vec{a} \in \C^g$.
Since $f_{g}(t; \vec{a}) \equiv f(t; {\vec{a}})$ in
$\C[[t]]^*/\C[[t^2, t^{2g+1}]]^*$,
the invertible sheaf $L:=\cO(-\divi(f_{g}(t; \vec{a}))) \in J_g$ 
corresponds to $\vec{a} \in \C^g$ by \eqref{jacobian}.
\end{enumerate}
\end{definition}

In order to study the structure of $J_g$,
we need to introduce some definitions.

\begin{definition}
For a natural number $k$, we define the {Abel-Jacobi map}
\[ \aj_{g, k}: \Sym^{k} (C_g \setminus \{ O \}) \to J_g  \]
by $\aj_{g, k}(D) := \cO(D - k [\infty])$.
\end{definition}

\begin{definition}\label{def:X}
We define a $(g \times 2g)$-matrix
\[ X_{2g} := 
\begin{pmatrix}
\chi_1 & \chi_0 & 0 & 0 & 0 & 0 &  \cdots & 0 \\
\chi_3 & \chi_2 & \chi_1 & \chi_0 & 0 & 0 &  \cdots & 0 \\
\chi_5 & \chi_4 & \chi_3 & \chi_2 & \chi_1 & \chi_0  & \cdots & 0 \\
\vdots &        &        &        &        &        & &\vdots & \\
\chi_{2g-1} & \chi_{2g-2} & \chi_{2g-3} & \chi_{2g-4} & 
\chi_{2g-5} & \cdots 
& \chi_1 & \chi_0
\end{pmatrix}
\]
with entries in $\C[a_1, \dots, a_g]$.
For $0 \leqslant k \leqslant 2g$,
let $X_{k}$
be the $(g \times k)$-submatrix of $X_{2g}$
consisting of the left $k$ columns of $X_{2g}$.
\end{definition}

\begin{lemma}\label{keylemma}
Let $k \in \Z$ and $\vec{a}=\C^g$.
Recall that $L=\cO(-\divi(f_g(t; \vec{a})))$
is the corresponding invertible sheaf.
\begin{enumerate}
\item
Assume that $0 \leqslant k \leqslant 2g-1$.
Then the following conditions are equivalent:
\begin{enumerate}
\item $h^0(C_g, L(k)) \not= 0$.
\item 
There exists $h(t) \in \C[t]\setminus\{0\}$
such that $\deg h(t) \leqslant k$ and $f_g(t; \vec{a})h(t) \in R_g$.
\item
There exists $\vec{b}=(b_i)_{i=0}^{k} \in \C^{k+1}$
such that $X_{k+1} \vec{b} = 0$ and $\vec{b} \not= 0$.
\end{enumerate}
\item
Assume that $0 \leqslant k \leqslant 2g-1$.
Then the following conditions are equivalent:
\begin{enumerate}
\item $L$ is in the image of $\aj_{g, k}$.
\item 
There exists $h(t) \in \C[t]$ 
such that $\deg h(t) \leqslant k$ and $f_{g}(t; \vec{a})h(t) \in R_g^*$.
\item
There exists $\vec{b}=(b_i)_{i=0}^{k} \in \C^{k+1}$
such that $X_{k+1} \vec{b} = 0$ and $b_0 \not= 0$.
\end{enumerate}
\end{enumerate}
\end{lemma}

\begin{proof}
First we prove (1):
\begin{align*}
h^0(C_g, L(k)) \not= 0
\Leftrightarrow&
\exists r \in R_g \setminus \{ 0 \}, ~
\divi(r) - \divi(f_{g}(t; \vec{a})) + k \cdot \infty \geqslant 0
\\
\Leftrightarrow&
\exists h (= \frac{r}{f_{g}(t; \vec{a})}) ~\in \C[t] \setminus \{ 0 \},~
\text{s.t.}
\deg h(t) \leqslant k, ~ f_{g}(t; \vec{a}) h(t) \in R_g
\\
\Leftrightarrow&
\exists 
h(t) = \sum_{j=0}^{k} b_j t^j \not= 0
~\text{s.t.}~
f_{g}(t; \vec{a}) h(t) = \sum_n (\sum_j b_j \chi_{n-j}) t^n \in R_g
\\
\Leftrightarrow&
\exists \vec{b} = (b_j) \in \C^{k+1} \setminus \{ 0 \}
~\text{s.t.}~
\sum_j b_j \chi_{n-j} = 0 ~(n=1, 3, \dots, 2g-1)
\\
\Leftrightarrow&
\exists \vec{b} = (b_j) \in \C^{k+1} \setminus \{ 0 \}
~\text{s.t.}~
X_{k+1} \vec{b} = 0.
\end{align*}
Next we prove (2).
If $E = \sum_{i=1}^{k} [\pi_g(\alpha_i)] 
\in \Sym^{k}(C_g \setminus \{ O \})$
with $\alpha_i \in \mathbb{P}^1 \setminus \{ O \}$,
then $\aj_{g, k}(E)$
is represented by 
$h(t)^{-1} \in S^*$ where 
$h(t) = \prod_{i=1}^{k} (1 - \frac{t}{\alpha_i})$.
(Here the factors with $\alpha_i= \infty$ are regarded as $1$.
Hence $h(t)$ is a polynomial of degree $\leqslant k$.)
Since $\vec{a}$ is represented by
$f_{g}(t; \vec{a}) \in S^*$,
we have 
$\vec{a} = \aj_{g, k}(E) \Leftrightarrow 
f_{g}(t; \vec{a}) h(t) \in R_g^*$.
This proves (a) $\Leftrightarrow$ (b).
The equivalence of (b) and (c) 
is seen in the same way as (1).
\end{proof}

The above lemma justifies the following definition.

\begin{definition}
We define the {\it theta divisor} $\Theta_g$ to be the zero locus of the polynomial $\det(X_{g}) \in \C[a_1, \dots,
a_g]$ in $\C^g$.  This is a divisor on $\C^g$, but we identify it with a divisor on $J_g$ via the isomorphism
\eqref{jacobian}, which is the same as $\Theta_g(x^{2g+1})$, defined in (\ref{def:theta_1}).
%appearing in Theorem~\ref{thm:beauville}.
\end{definition}

\begin{corollary}\label{keycor}
\begin{enumerate}
\item 
For $L \in J_g$, the following conditions are equivalent:
\[ (a) ~L \in \Theta_g,
\qquad
  (b) ~h^0(C_g, L(g-1) ) \not= 0,
\qquad
  (c) ~\det (X_{g}) = 0.
\]
\item
We have $\Image(\aj_{g, g-1}) \subset \Theta_g$.
However, 
$\Image(\aj_{g, g-1}) \not= \Theta_g$
if $g \geqslant 3$.
\item
For any $L \in J_g$ and $k \geqslant g$, 
the equivalent conditions
in Lemma \ref{keylemma} (1) hold.
However,
$\Image(\aj_{g, g}) \not= J_g$ if $g \geqslant 2$.
\item
The image of $\aj_{g, g+1}$ contains $J_g \setminus \Theta_g$.
\end{enumerate}
\end{corollary}
\begin{proof}
(1)
When $k=g-1$, the condition 
of Lemma \ref{keylemma} (1-c) is rephrased as (c),
which proves $(b) \Leftrightarrow (c)$.
The equivalence of (a) and (c) is the definition of $\Theta_g$.

(2) The first statement
follows from a trivial fact $R_g^* \subset R_g$.
The second statement is an effect of 
non-zero elements of $R_g \setminus R_g^*$.
A concrete example is given by
$g=3$ and $L = \cO(-\divi(1-t^5))$.
(See the last line in this proof for the case $g=1, 2$.)

(3) The first statement follows from the Riemann-Roch theorem.
The second statement is an effect of elements of $R_g \setminus R_g^*$.
A concrete example is given by $g=2$ and $L = \cO(-\divi(1-t^3))$.
(See the last line in this proof for the case $g=1$.)

(4)
We take $L \in J_g \setminus \Theta_g$.
Then $X_{g}$ is of rank $g$ by (1).
Hence $X_{g+2}$ is also of rank $g$,
and the linear equation
\[ (*) \qquad X_{g+2} \vec{b} = 0 \]
has two independent solutions 
$\vec{b} = (b_i)_{i=0}^{g+1} \in \C^{g+2}$.
By Lemma \ref{keylemma} (2),
it is enough to show
that (at least) one of these two solutions 
satisfies $b_0 \not= 0$.
We suppose that there exist
two independent solutions to $(*)$ with $b_0=0$.
Because the first row of $(*)$ reads 
$a_1 b_0 + b_1 = 0$, we have $b_1=0$ as well.
Let $Y$ be the
lower-right $((g-1) \times g)$-submatrix of $X_{g+2}$.
(This is to say $Y$ is constructed by removing 
the top row and the left two columns from $X_{g+2}$.)
Then $Y \vec{c} = 0$ has two independent solutions.
Hence $Y$ is of rank $\leqslant g-2$.
However, 
since $Y$ is the same as a submatrix of $X_{g}$
(obtained by removing the bottom row from $X_{g}$)
this implies $X_{g}$ is of rank $\leqslant g-1$.
This contradicts the fact that the rank of $X_g$ is $g$.
%the assumption $L \not\in \Theta_g$.

The implication $b_0 = 0 \Rightarrow b_1=0$
shows that the implication (1-c) $\Rightarrow$ (2-c) holds
when $k=2$ in Lemma \ref{keylemma}.
This explains why 
the equality holds in (2) for $g=1,2$, and in (3) for $g=1$.
\end{proof}

The following lemma gives 
an explicit formula for the map \eqref{jacobian},
composed with $\aj_{g, g}$,
restricted to an open dense subset:
\begin{lemma}\label{explicit_isom}
The composition
$\Sym^g(C_g \setminus \{O, \infty\}) 
\overset{\aj_{g, g}}{\to} J_g \cong \C^g$
is described as follows:
\[ D = \sum_{k=1}^g [\pi_g(\alpha_k)] \mapsto
  \vec{a} = (\frac{1}{2i-1} \sum_{k=1}^g \alpha_k^{-(2i-1)})_{i=1}^g
\qquad (\alpha_k \in \mathbb{P}^1 \setminus \{ O, \infty \}).
\]
\end{lemma}
\begin{proof}
We claim that the formal power series defined by $\Xi(t) := (1-t)\exp(\sum_{j=1}^{\infty} \frac{t^{2j-1}}{2j-1})$
belongs to $\C[[t^2]]$.  Indeed, we have
%\begin{align*}
%\Xi(t) &= 
%\frac{1-t^2}{1+t}\exp(\sum_{j=1}^{\infty} \frac{t^{2j-1}}{2j-1})
%= 
%(1-t^2)\exp(-\log(1+t) + \sum_{j=1}^{\infty} \frac{t^{2j-1}}{2j-1})
%\\
%&= 
%(1-t^2)\exp(\sum_{j=1}^{\infty} \frac{(-1)^{j}t^{j}}{j}
%+\sum_{j=1}^{\infty} \frac{t^{2j-1}}{2j-1})
%= 
%(1-t^2)\exp(\sum_{j=1}^{\infty} \frac{t^{2j}}{2j}),
%\end{align*}
%
\begin{align*}
\Xi(t) &= 
\exp(\log(1-t) + \sum_{j=1}^{\infty} \frac{t^{2j-1}}{2j-1})
= 
\exp(-\sum_{i=1}^{\infty} \frac{t^{i}}{i}
+\sum_{j=1}^{\infty} \frac{t^{2j-1}}{2j-1})
= \exp(-\sum_{i=1}^{\infty} \frac{t^{2i}}{2i}),
\end{align*}
which belongs to $\C[[t^2]]$.  Consequently, we have
\[ \Xi_g(t) := (1-t) \exp(\sum_{j=1}^g \frac{t^{2j-1}}{2j-1}) 
\in \C[[t^2, t^{2g+1}]].
\]
By replacing $t$ by $\alpha_k^{-1}t$ for $k=1, \dots, g$ and taking a product, we get
\[ 
\prod_{k=1}^g \Xi_g(\alpha_k^{-1}t)
=\prod_{k=1}^g(1-\alpha_k^{-1}t) 
   \exp(\sum_{i=1}^g \frac{(\alpha_k^{-1}t)^{2i-1}}{2i-1})
=h(t) f(t; \vec{a})
\in \C[[t^2, t^{2g+1}]],
\]
where $h(t) = \prod_{k=1}^g(1-\alpha_k^{-1}t)$.
Now we take 
$D := \sum_{k=1}^g [\pi_g(\alpha_k)] 
\in \Sym^g(C_g \setminus \{O, \infty\})$.
Then $\aj_{g, g}(D)$ is represented by $h(t)^{-1}$
in $\C[[t]]^*/\C[[t^2, t^{2g+1}]]^*$.
The above calculation shows that,
in $\C[[t]]^*/\C[[t^2, t^{2g+1}]]^*$,
the class of $h(t)^{-1}$ is the same as $f(t; \vec{a})$,
which represents $\vec{a}$.
This completes the proof.
\end{proof}

%%%%%%%%%%%%%%%%%%%%%%%%%%%%%%%%%%%%%%%%%%%%%%%%%%%%%%%%%
\subsection{Compactification of the Generalized Jacobian}
%%%%%%%%%%%%%%%%%%%%%%%%%%%%%%%%%%%%%%%%%%%%%%%%%%%%%%%%%
We write $\bar{J}_g$ for the compactified Jacobian of $C_g$
which parametrizes isomorphism classes of
torsion free $\cO_{C_g}$-modules $\mathcal{L}$ 
of rank one such that 
$h^0(C_g, \mathcal{L})-h^1(C_g, \mathcal{L})=1-g$
(see \cite{dsouza, rego}).
We have a natural inclusion $J_g \subset \bar{J}_g$,
by which we regard $J_g$ as
a Zariski dense open subscheme in $\bar{J}_g$ 
We also define 
$\bar{\Theta}_g = 
\{ \mathcal{L} \in \bar{J}_g ~|~ 
  h^0(C_g, \mathcal{L}(g-1)) \not= 0 \}.$
Similarly to $\Theta_g$,
we see that
$\bar{\Theta}_g$ is the same as 
$\bar{\Theta}_g(x^{2g+1})$,
defined in (\ref{def:theta_2}).
%appearing in Theorem \ref{thm:beauville}.

The normalization $\pi_g: \mathbb{P}^1 \to C_g$ factors as
$\mathbb{P}^1 \overset{\pi_k}{\to} C_k \overset{\pi_{k,g}}{\to} C_g$
for $k= 1, \dots, g$.
Explicitly, $\pi_{k, g}$ is given by
$\pi_{k, g}(x,y) = (x, x^{g-k}y)$.
We have a push-forward
$(\pi_{k, g})_*: \bar{J}_{k} \to \bar{J}_g$.
We also have an action of $J_g$ on $\bar{J}_g$
defined by $L \cdot \mathcal{L} = L \otimes \mathcal{L}$
for $L \in J_g$ and $\mathcal{L} \in \bar{J}_g$.

\begin{lemma}\label{cpt}
Let $k \in \{1, \dots, g \}$.
\begin{enumerate}
\item
The push-forward defines an isomorphism
$(\pi_{g-1, g})_*: \bar{J}_{g-1} \to \bar{J}_g \setminus J_g$.
\item
For any $L \in J_{g}$ and $\mathcal{L} \in \bar{J}_k$,
we have
$(\pi_{k, g})_*((\pi_{k, g})^*L \cdot \mathcal{L}) 
=L \cdot (\pi_{k, g})_* \mathcal{L}$
\item
We have a commutative diagram of algebraic groups
\[
\begin{matrix}
\C^g & \overset{\eqref{jacobian}}{\cong} & J_g \qquad \\[1mm]
\downarrow & & \downarrow_{(\pi_{k, g})^*} \\
\C^{k} & \overset{\eqref{jacobian}}{\cong} & J_{k} \qquad
\end{matrix}
\]
where the left vertical map is defined by
$(a_i)_{i=1}^{g} \mapsto (a_i)_{i=1}^{k}$.
\item
For any $\mathcal{L} \in \bar{J}_k$,
we have $\mathcal{L} \in \bar{\Theta}_k$ if and only if 
$(\pi_{k, g})_* \mathcal{L} \in \bar{\Theta}_g$.
\end{enumerate}
\end{lemma}
\begin{proof}
It is proved in \cite[Lemma 3.1]{Beauville99}  that
$(\pi_{k, g})_*: \bar{J}_k \to \bar{J}_g$ is a closed embedding.
Now (1) follows by induction from the elementary fact that 
any torsion-free $R_g$-submodule $M$ of rank one in $\C(C_g)=\C(t)$
satisfies $f(t) M = R_k$ for some $f(t) \in \C(t)^*$
and $k=0, 1, \dots, g$. (We set $R_0 =S$ by convention.)
(2) is a direct consequence of the projection formula.
(3) follows from the description of the isomorphism \eqref{jacobian}.
Since $\pi_{k, g}$ is a finite map, we have
$h^0(C_k, \mathcal{L}) = h^0(C_g, (\pi_{k, g})_* \mathcal{L})$,
which proves (4).
\end{proof}

%%%%%%%%%%%%%%%%%%%%%%%%%%%%%%%%%%%%%%%%%%%%%%%%%%%%%%%%%%%%%%%%%%%%%%%%%
\section{Singular fiber of the Mumford system with additive degeneration}
%%%%%%%%%%%%%%%%%%%%%%%%%%%%%%%%%%%%%%%%%%%%%%%%%%%%%%%%%%%%%%%%%%%%%%%%%

We use the notations of \S 2.  We apply the results of the previous section to study the level set
$M_g(0):=M_g(x^{2g+1})$ of the genus $g$ Mumford system.

%%%%%%%%%%%%%%%%%%%%%%%%%%%%%%%%%%%%%%%%%%%%%%%%%%%%%%%%%%%
\subsection{Matrix realization of the generalized Jacobian}
%%%%%%%%%%%%%%%%%%%%%%%%%%%%%%%%%%%%%%%%%%%%%%%%%%%%%%%%%%%
Let us take $h(x) := x^{2g+1} \in \Hg$.
Then the spectral curve $C_g(h)$ is precisely $C_g$
considered in the previous section.
We write $M_g(0)$ and $M_g(0)_{reg}$ for $M_g(h)$
and $M_g(h)_{reg}$.
We define a map
\[ i_{g}: M_{g-1}(0) \to M_g(0) 
\qquad i_{g}(l(x))=x l(x). 
\]
\begin{lemma}
Let $l(x) \in M_g(0)$.
Then $l(x)$ is in $M_g(0)_{reg}$
iff $l(0) \not= 0.$
In other words,
we have $M_g(0)_{reg} 
= M_g(0) \setminus i_{g}(M_{g-1}(0))$.
\end{lemma}
\begin{proof}
We first remark that
a traceless $2$ by $2$ matrix $A$ is regular
iff $A \not= 0$.
Hence $l(x) \in M_g(0)_{reg}$
iff $l(c) \not= 0$ for all $c \in \C$.
If $l(c) = 0$ for some $c \in \C$,
then $x^{2g+1}=-\det l(x)$ is divisible by $x-c$,
thus $c$ must be $0$.
\end{proof}

Combined with Theorem \ref{thm:beauville},
we obtain
\begin{theorem}\label{cor_beauville}
There exist isomorphisms
\[
 \bar{\phi}_g: M_g(0) 
    \cong \bar{J}_g \setminus \bar{\Theta}_g
\qquad \text{and} \qquad
 \phi_g: M_g(0) \setminus i_g(M_{g-1}(0)) 
   \cong J_g \setminus \Theta_g. 
\]
\end{theorem}

\begin{remark}\label{rem:phi}
We give an explicit description of $\phi_g$.
(Compare with Remark \ref{MandB}.)
Take $l(x) \in M_g(0) \setminus i_g(M_{g-1}(0))$.
Because of the relation $u_0 w_0 + v_0^2 = 0$,
we have $u_0 \not= 0$ or $w_0 \not= 0$.
In the first case,
$l(x)$ is mapped to the invertible sheaf
corresponding to the divisor $\sum_{i=1}^g [\alpha_i] - g [\infty]$,
where $\alpha_i = v(x_i)/x_i^g$ with $u(x) = \prod_{i=1}^g (x-x_i)$.
In the second case,
$l(x)$ is mapped to the invertible sheaf
corresponding to the divisor 
$\sum_{j=1}^{g+1} [-\beta_j] - (g+1) [\infty]$,
where $\beta_j = v(x_j)/x_j^g$ with 
$w(x) = \prod_{j=1}^{g+1} (x-x_j)$.
Note that for $l(x)$ with $u_0 w_0 \not= 0$,
the two definitions give the same divisor class.
Indeed,
one has
\[ \sum_i [\alpha_i] - g[\infty] 
\equiv
   - \sum_j [\beta_j] + (g+1)[\infty] 
\equiv
   \sum_j [-\beta_j] - (g+1)[\infty],
\]
where the first equivalence is seen by
$\divi(t^{2g+1} - v(t^2)) 
= \sum_i [\alpha_i] + \sum_j [\beta_j] - (2g+1) [\infty]$,
while the second follows from 
$\divi(1 - \frac{t^2}{\gamma^2}) = [\gamma] + [-\gamma] - 2 [\infty]$
for any $\gamma \in \C \setminus \{ 0 \}$.
We shall consider
the inverse map of $\phi_g$ in \S \ref{inverse}.
\end{remark}

%%%%%%%%%%%%%%%%%%%%%%%%%%%%%%%%%%%%%
\subsection{The Hamiltonian vector fields}
%%%%%%%%%%%%%%%%%%%%%%%%%%%%%%%%%%%%%
\begin{theorem}\label{main}
The vector fields $\X_1, \dots, \X_g$ on $M_g(0)$ are
linearized by the isomorphism $\bar{\phi}_g$
to the vector fields induced 
by the action of $J_g$ on $\bar{J}_g$.
More precisely, we have the following:
\begin{enumerate}
\item
For any $i=1, \dots, g$,
the vector fields $\X_i$
on $\Mg(0) \setminus i_g(M_{g-1}(0))$
are mapped to (the restriction of)
the invariant vector fields $\frac{\pt}{\pt a_{i}}$ on $\C^g$
by the isomorphism 
$\phi_g$
in Theorem \ref{cor_beauville}
composed with \eqref{jacobian}.
\item
The map $i_{g}: M_{g-1}(0) \to M_{g}(0)$
is a closed embedding.
The vector fields $\X_1, \dots, \X_{g-1}$ on $M_{g-1}(0)$
are mapped to $\X_1, \dots, \X_{g-1}$ on $M_{g}(0)$ by 
$i_{g}$,
while the vector field $\X_g$ is zero on $i_{g}(M_{g-1}(0))$.
\item
The level set $M_g(0)$ is stratified by
$g+1$ smooth affine varieties,
which are invariant for the flows of the vector fields
$\X_1, \dots, \X_g$;
they are isomorphic to $\C^{k} \setminus \Theta_k$
for $k=0, \dots, g$.
\end{enumerate}
\end{theorem}
\noindent
{\it Proof.}
(2) follows from (1) and Lemma \ref{cpt}.
(3) is a consequence of (1) and (2).
We prove~(1).
Since the vector fields in question are all holomorphic,
it suffices to show this assertion
on some open dense subset.
We define 
\[ \mathcal{S}' : = \{ \sum_{i=1}^g [(x_i, y_i)]
  \in \Sym^g (C_g \setminus \{ O, \infty \})
    ~|~
  x_i \not= x_j ~\text{for all}~ i \not= j \}.
\]
\begin{lemma}
  The map $\aj_{g, g}$, restricted to $\mathcal{S}'$, is an open immersion whose image is a dense open subset of
  $J_g$.
\end{lemma}
\begin{proof}
Suppose that $\sum_{i=1}^g [(x_i, y_i)]$ and $\sum_{i=1}^g [(x'_i, y'_i)]$ 
have the same image in $J_g$.  If we set $\alpha_i =
y_i/x_i^g,~ \alpha'_i=y'_i/{x'}_i^g$, then this amounts to saying that $f(t)=\prod_{i=1}^g (1-\alpha_i^{-1}
t)(1+{\alpha'}_i^{-1} t)$ is in $R_g$.  Since $f(t)$ is of degree $2g$, we must have $f(t)=f(-t)$.  This implies
$\sum_{i=1}^g [\alpha_i] = \sum_{i=1}^g [\alpha_i']$ 
(in $\Sym^g(\mathbb{P}^1 \setminus \{O, \infty\})$) by the definition of
$\mathcal{S}'$, and the injectivity follows.  The rest follows from Lemma \ref{keylemma} (2).
\end{proof}

Now we consider the vector fields on $\mathcal{S}'$.
Since the computation made in the proof of Theorem \ref{thm:mumford}
is valid in this situation, 
it follows by putting 
$x_i = \alpha_i^2, y_i = \alpha_i^{2g+1}$
in \eqref{vf-xy}
that,
with local coordinates $\alpha_i$,
the vector fields $\X_i$ are expressed by
  \begin{align}\label{vf-alpha}
    \begin{pmatrix}
      1 & 1 & \cdots & 1 \\ 
      \alpha_1^2 & \alpha_2^2 & \cdots & \alpha_g^2 \\
      \vdots & & & \vdots \\
      \alpha_1^{2g-2} & \alpha_2^{2g-2} & \cdots & \alpha_g^{2g-2} 
    \end{pmatrix}
    \begin{pmatrix}
      \frac{\X_g \alpha_1}{\alpha_1^{2g}} & 
      \frac{\X_{g-1} \alpha_1}{\alpha_1^{2g}} & \cdots &
      \frac{\X_1 \alpha_1}{\alpha_1^{2g}} \\
      \frac{\X_g \alpha_2}{\alpha_2^{2g}} & \cdots & 
      & \frac{\X_1 \alpha_2}{\alpha_2^{2g}} \\
      \vdots & & & \vdots \\
      \frac{\X_g \alpha_g}{\alpha_g^{2g}} & \cdots & 
      & \frac{\X_1 \alpha_g}{\alpha_g^{2g}} 
    \end{pmatrix}
    =
    - \mathbb{I}_g.
  \end{align}  
%Now Theorem \ref{main} (1) follows from \eqref{vf-alpha} and Lemma \ref{explicit_isom}.  
Using Lemma \ref{explicit_isom} and \eqref{vf-alpha} one computes that $\X_k a_i=\delta_{i,k},$ for $1\leqslant
i,k\leqslant g$, which leads to (1) in Theorem \ref{main}.

%%%%%%%%%%%%%%%%%%%%%%%%%%%%%%%%%%%%%%%%%%%%%%%%%%%%%
\section{Rational solution to the Mumford system}
%%%%%%%%%%%%%%%%%%%%%%%%%%%%%%%%%%%%%%%%%%%%%%%%%%%%%
%
In view of Theorem \ref{main},
an explicit description of the inverse map 
$$
\phi_g^{-1}: 
  J_g \setminus \Theta_g \to M_g(0) \setminus i_{g}(M_{g-1}(0))
  ;\quad \vec{a} \mapsto ~\begin{pmatrix} v(x) & w(x) \\ u(x) & -v(x)
                      \end{pmatrix}.
$$
of $\phi_g$
gives rise to a rational solution to the Mumford system.
This will be done in \S \ref{inverse},
then we present a concrete algorithm to compute
rational solutions in \S \ref{algorithm}.

%%%%%%%%%%%%%%%%%%%%%%%%%%%%%%%%%%%%%%%%%%%%%%%%%
\subsection{The map $\phi_g^{-1}$}\label{inverse}
%%%%%%%%%%%%%%%%%%%%%%%%%%%%%%%%%%%%%%%%%%%%%%%%%
We introduce some notations.
For $\vec{a}=(a_1,\cdots,a_g) \in \C^g$, 
let $\bar{X} = \bar{X}(\vec{a})$ be the $2g$ by $g$ matrix:
\begin{align*}
  \bar{X} =
  \begin{pmatrix}
    \chi_0 & 0 & \cdots \\
    \chi_1 & 0 & \cdots \\
    \chi_2 & \chi_0 & 0 & \cdots \\
    \vdots \\
    \chi_g & \chi_{g-2} & \cdots \\
    \chi_{g+1} & \chi_{g-1} & \cdots \\
    \vdots \\
    \chi_{2g-2} & \chi_{2g-4} & \cdots & \chi_2 & \chi_0\\    
    \chi_{2g-1} & \chi_{2g-3} & \cdots & \chi_3 & \chi_1\\    
  \end{pmatrix},
\end{align*}
where $\chi_i = \chi_i(\vec{a})$ are given by Definition \ref{def:chi} (1).  We write $\bar{X}_g =
\bar{X}_g(\vec{a})$ for the submatrix consisting of the last $g$ rows of $\bar{X}$.  We remark that $\bar{X}$ and
$X_{2g}$ of Definition \ref{def:X} are closely related (for instance we have $\det \bar{X}_g = \det X_{g}$), but
they come out of different contexts, and it seems more natural to use both of them.  We divide $\bar{X}$ into a
$g+1$ by $g$ matrix $A=A(\vec{a})$, a $g-1$ by $g-1$ matrix $B=B(\vec{a})$ and a vector
$\vec{\phi}=\vec{\phi}(\vec{a}) = \,^t (\phi_1,\cdots,\phi_{g-1})$:
\begin{align*}
  \begin{split}
    &A_{i,j} = \bar{X}_{i,j} = \chi_{i-2j+1}
~\text{ for } 1 \leqslant i \leqslant g+1, ~ 1 \leqslant j \leqslant g, \\
    &B_{i,j} = \bar{X}_{g+1+i,1+j} = \chi_{i-2j+g}
~\text{ for } 1 \leqslant i,j \leqslant g-1,\\
    &\phi_i = \bar{X}_{g+1+i,1} =  \chi_{i+g}
~\text{ for } 1 \leqslant i \leqslant g-1.
  \end{split}
\end{align*} 
Let $\tau_g = \tau_g(\vec{a})$ be 
the polynomial function on $\C^g$ given by
\begin{align}\label{tau-rKdV}
  \tau_g(\vec{a}) = \det \bar{X}_g(\vec{a}) (=\det X_{g}(\vec{a}) ).
\end{align}
Note that $\tau_g$ is essentially the Schur function associated to the partition $\nu = (g, g-1, \cdots, 1)$. 
(See (3) in the proof of Proposition \ref{prop:KdV}.)
Recall that the $g$ vector fields $\X_i$ on $M_g(0)$ induce the translation invariant vector fields $\X_i =
\frac{\partial}{\partial a_{i}}$ on $\C^g$ (Theorem \ref{main} (1)).  For a rational function $s \in
\C(a_1,\cdots,a_g)$ we write $s' := \frac{\partial}{\partial a_1}s$ and $s^{(k)} := \frac{\partial^k}{\partial
a_1^k} s$ for $k=1, 2, \dots$.

Let $U$ be the open subset of 
$J_g \setminus \Theta_g = 
\{\vec{a} \in \C^g ~|~ \tau_g(\vec{a}) \not= 0 \}$
defined by
\[
  U := \{ \vec{a} \in \C^g ~|~ \det B(\vec{a}) \neq 0 \text{ and } 
                         \tau_g(\vec{a}) \neq 0 \}.
\]
The next proposition is 
a key to an explicit formula for $\phi_g^{-1}$:
\begin{proposition}\label{defofpq}
Suppose $\vec{a} \in U$.
We denote by 
$p(t;\vec{a}) = \sum_{k=0}^g p_k t^k$
the polynomial, whose coefficients are defined by
\begin{align}\label{a-p}
  \vec{p} = \,^t(p_0,p_1,\cdots,p_g) 
          := A(\vec{a}) 
   \begin{pmatrix} 1 & 0 \\ 0 & -B(\vec{a})^{-1} \end{pmatrix}
   \begin{pmatrix} 1 \\ \vec{\phi}(\vec{a}) \end{pmatrix}.
\end{align}
Then $p(t;a)$ is the unique polynomial of degree at most $g$,
which satisfies $p_0 = 1$ and
\begin{align}\label{chi-h}    
    \sum_{k=0}^{2g-1} \chi_k t^k
    \equiv p(t;\vec{a}) 
    \text{ in } \C[[t]]^\ast / \C[[t^2,t^{2g+1}]]^\ast.
  \end{align} 
\end{proposition}
\begin{proof}
We see that \eqref{chi-h} with $p_0=1$
is equivalent to the existence of a polynomial 
$b(t) = 1+ \sum_{j=1}^{g-1} b_j t^{2j}$ such that
\begin{align}\label{b-h}
  p(t;\vec{a}) \equiv (\sum_{k=0}^{2g-1} \chi_k t^k) \cdot b(t)
  \mod t^{2g} \C[[t]].
\end{align}
Then we have
\begin{align*}
   \eqref{b-h} ~\Leftrightarrow~
   \bar{X} \vec{b} = \begin{pmatrix} 
                        \vec{p} \\ 0 \\ \vdots \\ 0 \end{pmatrix} 
   ~\Leftrightarrow \qquad~
   \begin{cases}
    \mathrm{(ia)} ~ A \vec{b} = \vec{p} \\
    \mathrm{(ib)} ~ (\vec{\phi} \, B) \vec{b} = \,^t (0,\cdots,0)
   \end{cases} 
\end{align*}
where $\vec{b} = \,^t(1,b_1,b_2,\cdots,b_{g-1})$.  
When $\det B \neq 0$, 
(ib) has the unique solution
$$
  \vec{b} 
= \begin{pmatrix} 1 & 0 \\ 0 & -B^{-1} \end{pmatrix}
  \begin{pmatrix} 1 \\  \vec{\phi} \end{pmatrix},
$$
with which
(ia) is equivalent to \eqref{a-p}.  
This completes the proof.
\end{proof}

\begin{theorem}\label{lemma:phi^-1}
If $\vec{a} \in U$, then 
  $\phi_g^{-1}(\vec{a}) =
  \begin{pmatrix} v(x) & w(x) \\ u(x) & -v(x) \end{pmatrix}$ 
  is given by  
  \begin{align}\label{a-u}
    &u(t^2) = \frac{(-1)^g}{p_g(\vec{a})^2} p(t;\vec{a}) p(-t;\vec{a}), \\
    &v(x) = \frac{1}{2} \frac{\partial}{\partial a_1} u(x), \qquad
    w(x) =  (x-2 u_{g-1}) u(x) -\frac{1}{2}
      \frac{\partial^2}{\partial a_1^2} u(x).
  \end{align}
\end{theorem}
For a proof, we need a few lemmas:
\begin{lemma}\label{lemma:d-chi}
  For $k=1,\cdots,g$.
  we have
  \begin{align*}
  \mathrm{(1)}~ \frac{\partial}{\partial a_{k}} \chi_j = \chi_{j-2k+1},
  \qquad 
  \mathrm{(2)}~ \frac{\partial}{\partial a_{k}} \chi_j 
                = \bigl(\frac{\partial}{\partial a_1}\bigr)^{2k-1} \chi_j.
  \end{align*}
\end{lemma}
\begin{proof}
By operating with $\frac{\partial}{\partial a_{k}}$ on \eqref{chi}, we obtain
$$ 
  t^{2k-1} \exp(\sum_{i=1}^{g} a_i t^{2i-1}) 
= \sum_{j=0}^{\infty} (\frac{\partial}{\partial a_{k}} \chi_j) t^j.
$$
Thus we get 
$\sum_{j=0}^{\infty} 
(\frac{\partial}{\partial a_{k}} \chi_j - \chi_{j-2k+1}) t^j = 0$, 
and (1) follows. 
(2) follows from (1).
\end{proof}

\begin{lemma}\label{lemma:p-tau}
  For $\vec{a} \in U$, we have the following:
  \begin{align*}
    \mathrm{(1)} ~\tau_g(\vec{a})  = p_g \det B,
    \qquad                 
    \mathrm{(2)} ~\tau'_g(\vec{a}) = p_{g-1} \det B, 
    \qquad 
    \mathrm{(3)} ~\tau''_g(\vec{a}) = 2 p_{g-2} \det B.
  \end{align*}
\end{lemma}
\begin{proof}
(1) Since $\vec{a} \in U$, we can write $B(\vec{a})^{-1} = \frac{1}{\det B(\vec{a})} \bar{B}(\vec{a})$ where
$\bar{B}(\vec{a})$ is the matrix of cofactors of $B(\vec{a})$.  We write $B_k$ for the $g-1$ by $g-2$ submatrix of
$B(\vec{a})$ obtained by removing the $k$-th column of $B(\vec{a})$.  We have
\begin{align*}
  &p_g \det B = \chi_g \det B - \sum_{k=1}^{[\frac{g}{2}]}\chi_{g-2k} 
                \sum_{j=1}^{g-1} \bar{B}_{k,j} \phi_j,\\
  &\tau_g = \det \bar{X}_g  
  = \chi_g \det B + \sum_{k=1}^{[\frac{g}{2}]} (-1)^k \chi_{g-2k} 
  \det \begin{pmatrix}\vec{\phi} B_k \end{pmatrix}.
\end{align*}
Now the claim follows from the following fact
\begin{align}\label{X-pB}
  \det \begin{pmatrix} \vec{\phi} B_k \end{pmatrix}=(-1)^{k-1} \sum_{j=1}^{g-1} \bar{B}_{k,j} \phi_j.
\end{align}
(2)
  Using Lemma \ref{lemma:d-chi}, we get
  \begin{align*}
    \tau_g'
    = \det \begin{pmatrix} 
             \chi_{g-1} & \chi_{g-3} & \chi_{g-5} & \cdots \\
             \vec{\phi} & & B
           \end{pmatrix}.
  \end{align*}
  On the other hand, we have
  \begin{align*}
    p_{g-1} \det B = \chi_{g-1} \det B 
                 - \sum_{k=1}^{[\frac{g-1}{2}]} 
                   \chi_{g-1-2k} \sum_{j=1}^{g-1} \bar{B}_{k,j} \phi_j,
  \end{align*}
  which coincides with $\tau_g'$ by \eqref{X-pB}.
  
\noindent  (3)  
  Using Lemma \ref{lemma:d-chi}, we have 
  \begin{align*}
    \tau_g''
    &= 
       \det \begin{pmatrix} 
             \chi_{g-2} & \chi_{g-4} & \chi_{g-6} & \cdots \\
             \chi_{g+1} & \chi_{g-1} & \chi_{g-3} & \cdots \\
             \chi_{g+2} & \chi_{g} & \chi_{g-2} & \cdots \\
             \vdots \\
             \chi_{2g-1} & \chi_{2g-3} & \chi_{2g-5} & \cdots \\ 
            \end{pmatrix}
           + 
       \det \begin{pmatrix} 
             \chi_{g-1} & \chi_{g-3} & \chi_{g-5} & \cdots \\
             \chi_{g} & \chi_{g-2} & \chi_{g-4} & \cdots \\
             \chi_{g+2} & \chi_{g} & \chi_{g-2} & \cdots \\
             \vdots \\
             \chi_{2g-1} & \chi_{2g-3} & \chi_{2g-5} & \cdots \\ 
            \end{pmatrix}.
\end{align*}
By using \eqref{X-pB}, 
the first term in r.h.s. turns out to be $p_{g-2} \det B$.  
On the other hand, it follows from 
the following lemma that the first and second terms coincide.
This completes the proof.
\end{proof}
\begin{lemma}
Let $X_0, \dots, X_{2g-1}$ be independent variables.
We define two elements in the polynomial ring $\C[X_0, \dots, X_{2g-1}]$:
\[
     Q_1:=  \det \begin{pmatrix} 
             X_{g-2} & X_{g-4} & X_{g-6} & \cdots \\
             X_{g+1} & X_{g-1} & X_{g-3} & \cdots \\
             X_{g+2} & X_{g} & X_{g-2} & \cdots \\
             \vdots \\
             X_{2g-1} & X_{2g-3} & X_{2g-5} & \cdots \\ 
            \end{pmatrix},
~~~
     Q_2:=  \det \begin{pmatrix} 
             X_{g-1} & X_{g-3} & X_{g-5} & \cdots \\
             X_{g} & X_{g-2} & X_{g-4} & \cdots \\
             X_{g+2} & X_{g} & X_{g-2} & \cdots \\
             \vdots \\
             X_{2g-1} & X_{2g-3} & X_{2g-5} & \cdots \\ 
            \end{pmatrix}.
\]
Then we have $Q_1=Q_2$.
\end{lemma}
\begin{proof}
We define a derivation $\partial$ 
on $\C[X_0, \dots, X_{2g-1}]$ by
$\partial X_j = X_{j-2}$ for $2 \leqslant j \leqslant 2g-1$
and $\pt X_0 = \pt X_1 = 0$.
We define 
\[ T :=
       \det \begin{pmatrix} 
             X_{g} & X_{g-2} &  \cdots & 0 \\
             X_{g+1} & X_{g-1} & \cdots & 0 \\
             X_{g+2} & X_{g} &  & \vdots \\
             \vdots & \vdots & & X_0 \\
             X_{2g-1} & X_{2g-3}  & \cdots &X_1\\ 
            \end{pmatrix}.
\]
We calculate $\partial T$ in two ways.
By differentiating columns, we see that
$\pt T = 0$ since $\pt X_0 = \pt X_1 = 0$.
By differentiating rows, we see that
$\pt T = Q_1-Q_2$.
This completes the proof.
\end{proof}

\noindent
{\it Proof of Theorem \ref{lemma:phi^-1}.}
{}From Lemma \ref{lemma:p-tau} (1) 
we have $p_g \neq 0$ on $U$, 
thus $p(t;\vec{a})$ is written as
$p(t;\vec{a}) = \prod_{j=1}^g (1- \frac{t}{\alpha_j})$
so that $p_g =
(-1)^g \prod_{j=1}^g \frac{1}{\alpha_j}$.  
Proposition \ref{defofpq} shows that
$u(x) = \prod_{j=1}^g (x-\alpha_j^2)$ 
(cf. Remark \ref{rem:phi}).   
Thus we have
  $$
    u(t^2) = \prod_{j=1}^g (t-\alpha_j)(t+\alpha_j)      
           = (\prod_{j=1}^g - \alpha_j^2) p(t;\vec{a}) p(-t;\vec{a}),
  $$
and \eqref{a-u} follows.  
The action of $\X_g$ \eqref{time-evol} on $M_g$ 
is written as follows:
\begin{align}\label{d_g-uvw}
  \begin{split}
    &\X_1 u(x) = 2 v(x), \\
    &\X_1 v(x) = - w(x) + (x-u_{g-1}+w_g) u(x), \\
    &\X_1 w(x) = 2 (x-u_{g-1}+w_g) v(x).
    \end{split}
\end{align}
To obtain $v(x)$ and $w(x)$, we use the first two equations, the relation $\X_1 =\frac{\partial}{\partial a_1}$
which comes from Theorem \ref{main} (1), and the fact that $w_g=-u_{g-1}$ on $M_g(0)$, as follows from
$u(x)w(x)+v(x)^2=x^{2g+1}$.
\qed

%%%%%%%%%%%%%%%%%%%%%%%%%%%%%%%%%%%%%%%%%%%%%%%%%%%%%
\subsection{Algorithm}\label{algorithm}
%%%%%%%%%%%%%%%%%%%%%%%%%%%%%%%%%%%%%%%%%%%%%%%%%%%%%
We present an explicit algorithm to compute 
a rational solution to the Mumford system.
This can be considered as a degenerate version of
\cite{McK-vM} (see also \cite[\S 10]{Mumford-Book}),
where a solution is given in terms of 
the hyperelliptic $\wp$-function.
The function $\rho_g$ defined in \eqref{def:rho} below
corresponds to a degenerate version of the hyperelliptic $\wp$-function.

\begin{definition}
We define a family of polynomials
$U_0, \dots, U_{g-1}, V_0, \dots, V_{g-1}, 
 W_0, \dots, W_{g} \in \C[T_0, \dots, T_{2g}]$
as follows.
We set
\[ U_{g-1}=T_0, \quad
   V_{g-1}=\frac{1}{2}T_1, \quad
   W_{g}=-T_0.
\]
Assume we have defined $U_{g-i}, V_{g-i}, W_{g-i+1}$
for $i = 1, \dots, k$.
Then we define
\begin{align*}
U_{g-k-1}
&= \frac{1}{4} \ddot{U}_{g-k}+U_{g-1}U_{g-k}
   -\frac{1}{2} \big(\sum_{j=g-k}^{g-1} U_j W_{2g-j-k}   + \sum_{j=g+1-k}^{g-1} V_j V_{2g-j-k} \big),\\
V_{g-k-1} &=  \frac{1}{2} \dot{U}_{g-k-1},\\
W_{g-k}\, \,
&= -\frac{1}{4} \ddot{U}_{g-k}-U_{g-1}U_{g-k}
 -\frac{1}{2} \big(\sum_{j=g-k}^{g-1} U_j W_{2g-j-k} +\sum_{j=g+1-k}^{g-1} V_j V_{2g-j-k} \big).
\end{align*}
Here $F \mapsto \dot{F}$ is the derivation
on $\C[T_0, \dots, T_{2g}]$
defined by
$\dot{T}_i=T_{i+1}$ for $i=0, 1, \dots, 2g-1$,
and by $\dot{T}_{2g}=0$.
The first examples of $U_k$ are given by
 \begin{align*}
    &U_{g-1} = T_0, \\
    &U_{g-2} = \frac{1}{4} T_2 + \frac{3}{2} T_0^2, \\
    &U_{g-3} = \frac{1}{16} T_4 + \frac{5}{8} T_1^2 
             + \frac{5}{4} T_0 T_2 + \frac{5}{2} T_3.
  \end{align*}
\end{definition}

\begin{proposition}\label{prop:uvw-tau}
  Let $\rho_g = \rho_g(\vec{a})$ be the rational function in $\C[a_1, \dots, a_g, \frac{1}{\tau_g}]$ given by
\begin{equation}\label{def:rho}
\rho_g(\vec{a}) = \frac{\partial^2}{\partial a_{1}^2}\log \tau_g(\vec{a}).
\end{equation}
Then, the functions
\[ u_k := U_k(\rho_g, \rho_g', \dots, \rho_g^{(2g)}),~~
   v_k := V_k(\rho_g, \rho_g', \dots, \rho_g^{(2g)}),~~
   w_k := W_k(\rho_g, \rho_g', \dots, \rho_g^{(2g)})
\]
give a rational solution for the genus $g$ Mumford system.
\end{proposition}
\begin{proof}
By using Theorem \ref{lemma:phi^-1}, when $\vec{a} \in U$, 
$u_{g-1}$ is written in terms of $p_j$ as
\begin{align*}%\label{u-p}
  u_{g-1} = \frac{2 p_{g-2} p_g - p_{g-1}^2}{p_g^2}.
\end{align*}
{}From Lemma \ref{lemma:p-tau}, this turns out to be
$$
  u_{g-1} = \frac{\tau_g'' \tau_g - (\tau_g')^2}{\tau_g^2} = \rho_g.
$$
Since $\rho_g$ has poles only on $\Theta_g$,
the domain of the solution of $u_{g-1}$ is extended 
from the open subset $U$ of $J_g \setminus \Theta_g$ to 
$J_g \setminus \Theta_g$. 

The first two equations of \eqref{d_g-uvw} yield 
\begin{align*}
  &v_{g-k} = \frac{1}{2} u_{g-k}', 
\\
  &\frac{1}{2} u_{g-k}'' = -w_{g-k}+u_{g-k-1}+(w_g-u_{g-1})u_{g-k}.
\end{align*}
If we look at the coefficient of
$x^{2g-k}$ in the equation $u(x)w(x) + v(x)^2 = x^{2g+1}$, we get 
\begin{align*}
  \sum_{j=g-k}^{g-1} u_j w_{2g-j-k} + u_{g-k-1} + w_{g-k} 
  + \sum_{j=g+1-k}^{g-1} v_j v_{2g-j-k} = 0.
\end{align*}
The proposition follows from these three equations.
\end{proof}

\begin{example}(Rational solution)
\\
(i) $g=2$ case:
$$
  \tau_2(\vec{a}) = \frac{a_1^3}{3} - a_2, \qquad 
  \rho_2(\vec{a}) = \frac{-3 a_1 (a_1^3 + 6 a_2)}{(a_1^3 - 3 a_2)^2}.
$$
(ii) $g=3$ case:
\begin{align*}
  &\tau_3(\vec{a}) = \frac{a_1^6}{45} - \frac{a_1^3 a_2}{3} - a_2^2 + a_1 a_3, \\
  &\rho_3(\vec{a}) = \frac{-3 (2 a_1^{10} + 675 a_1^4 a_2^2 - 1350 a_1 a_2^3 - 270 a_1^5 a_3 + 675 a_3^2)}
            {(a_1^6-15 a_1^3 a_2 - 45 a_2^2 + 45 a_1 a_3)^2}.
\end{align*}  
(iii) $g=4$ case:
\begin{align*}
  &\tau_4(\vec{a}) = \frac{a_1^{10}}{4725} - \frac{a_1^7 a_2}{105}
            - a_1 a_2^3 + \frac{a_1^5 a_3}{15} + a_1^2 a_2 a_3 - a_3^2
            - \frac{a_1^3 a_4}{3} + a_2 a_4.
\end{align*}  
\end{example}

%%%%%%%%%%%%%%%%%%%%%%%%%%%%%%%%%%%%%%%%%%%%%%%
\section{Relation to the KdV hierarchy}
%%%%%%%%%%%%%%%%%%%%%%%%%%%%%%%%%%%%%%%%%%%%%%%

As we already pointed out in the introduction, the Mumford systems and the KdV hierarchy are intimately related.
We briefly examine the relationship, with focus on the rational solutions.  Recall that the KdV hierarchy is
defined by the family of compatible Lax equations (see \cite{Jimbo-Miwa00, sw} and references therein):
$$
  \frac{\partial}{\partial x_{2i-1}} \mathcal{L} 
  =
  [\mathcal{L}^{i-\frac{1}{2}}_+ ~,~ \mathcal{L} ],
  ~~ \text{ for } i=1,2,3,\cdots.
$$
Here $\mathcal L$ is a differential operator of the form $\partial_x^2+f$ where $f$ is a function of $\vec{x} =
(x,x_1,x_3, \cdots) \in \C^\infty$ and $\pt_{x} f = \frac{\pt f}{\pt x} + f \cdot \pt_{x}$.  The square root
$\mathcal L^{\frac{1}{2}}$ is computed in the ring of formal pseudo-differential operators; the index $+$ in
$\mathcal{L}^{i-\frac{1}{2}}_+$ means that we take the differential part of $\mathcal{L}^{i-\frac{1}{2}}$.  The
first three equations ($i=1,2,3$) are given as follows:
\begin{align*}
  &\frac{\partial f}{\partial x_1} = \frac{\partial f}{\partial x},\\
  &\frac{\partial f}{\partial x_3} = \frac{1}{4} \frac{\partial^3 f}{\partial x^3} +\frac{3}{2} f \cdot \frac{\partial f}{\partial x},\\
  &\frac{\partial f}{\partial x_5} = \frac{1}{16}\frac{\partial^5f}{\partial x^5} +\frac{5}{8} f \cdot \frac{\partial^3f}{\partial x^3}
      + \frac{5}{4} \, \frac{\partial f}{\partial x}\cdot \frac{\partial^2f}{\partial x^2} +\frac{15}{8} f^2 \cdot \frac{\partial f}{\partial x}.
\end{align*}  
In the sequel, we identify $x$ with $x_1$, as suggested by the first equation of the above list and we consider the
rational solutions to the KdV hierarchy.  According to \cite{AMM77}, there is for every positive integer $g$ an
essentially unique solution, depending on $g$ parameters:
  \begin{enumerate}
  \item
  Suppose $f=f(x_1, x_3, \dots)$ be a non-zero rational function satisfying the KdV hierarchy.  Then there exist $g
  \in \Z_{>0}$ and $c_1 \in \C$ such that $f(x_1, 0, 0, \dots) = -\frac{g(g+1)}{(x_1-c_1)^2}$.  Moreover, $f$
  depends on the $g$ variables $x_1,x_3, \dots, x_{2g-1}$, and is independent of the other variables $x_{2i-1}$.
  In this case, we call $f$ a genus $g$ rational solution.
\item
  If $f$ and $\tilde{f}$ are genus $g$ rational solutions,
  then there exist $c_1, c_3, \dots, c_{2g-1} \in \C$
  such that 
$\tilde{f}(x_1 - c_1, \dots, x_{2g-1} - c_{2g-1}) 
= f(x_1, \dots, x_{2g-1})$.
\end{enumerate}

An explicit formula for these rational solutions is given in the following proposition, which is known in different
forms, as indicated in the proof below. The upshot, in connection with our result, is that the rational solutions
to the KdV hierarchy of genus $0,1,\dots,g$ fill up a very specific invariant manifold of the genus $g$ Mumford
system and form, combined, the complement of the completed theta divisor of the compactified Jacobian of the
singular curve $y^2=x^{2g+1}$. 

\begin{proposition}\label{prop:KdV}
  The function $f=2 \rho_g(\vec{a})$, defined in \eqref{def:rho}, gives a rational solution for the KdV hierarchy upon substituting
  \begin{align}\label{a-x}
    &a_i = x_{2i-1}, ~ \text{ for $i=1, \cdots, g$}.
  \end{align}
  This solution is non-trivial for the first $g-1$ vector fields $\frac{\partial}{\partial x_3},
  \frac{\partial}{\partial x_5}, \dots , \frac{\partial}{\partial x_{2g-1}}$ of the hierarchy, and trivial for the
  higher ones.
\end{proposition}
\begin{proof}
We sketch three different approaches to this result. 

(1) The KdV hierarchy is known to have Wronskian solutions, constructed as follows (See \cite{Hirota92} for details):
Fix $g\in \Z_{>0}$, and consider $g$ functions $f_1,\dots,f_g$ of 
$\vec{x}=(x_1,x_3,x_5,\cdots) \in \C^\infty$; for
$k\in\Z_{>0}$ we denote $f_i^{(k)} := \frac{\partial^k}{\partial x_1^k} f_i$. If these functions satisfy
\begin{align}\label{f-condition}
    \frac{\partial}{\partial x_{2k-1}} f_i 
    = \bigl(\frac{\partial}{\partial x_1}\bigr)^{2k-1} f_i,
      \qquad \text{for } k \in \Z_{>0},  
  \end{align}
then $2\frac{\partial^2}{\partial x_1^2} \log T(\vec x)$ satisfies the KdV hierarchy, where $T(\vec x)$ is defined
by
\begin{align}\label{tau-KP}
    T(\vec{x}) 
    :=
    \det \begin{pmatrix} 
           f_1 & f_1^{(1)} & \cdots & f_1^{(g-1)} \\
           f_2 & f_2^{(1)} & \cdots & f_2^{(g-1)} \\
           \vdots & \vdots & & \vdots\\
           f_g & f_g^{(1)} & \cdots & f_g^{(g-1)} 
         \end{pmatrix}.
\end{align}
In view of Lemma \ref{lemma:d-chi}, the functions $f_i:= \chi_{2g-2i+1}$ with $a_i=x_{2i-1}$ for $i=1,\cdots, g$,
satisfy~\eqref{f-condition}.  With this choice of functions,  $T(\vec{x})$ is precisely  $\tau_g(\vec{a})$, and the
result follows.

(2) In \cite[IIIa \S 10-11]{Mumford-Book}, Mumford shows, building upon the work \cite{McK-vM} of McKean-van
Moerbeke that a solution to the Mumford system, associated to an arbitrary \emph{smooth} hyperelliptic curve,
yields a solution to the KdV hierarchy.  In our case the hyperelliptic curve is not smooth, yet Mumford's argument
depends only on (differential) algebra, so we can construct as in the smooth case a rational solution $f$ to the
KdV equation from the rational solution which we constructed to the genus $g$ Mumford system.  Finally we obtain $f
= 2 u_{g-1}$, which leads precisely to the proposed solution.

(3) In the Grassmannian approach to the KdV equation \cite{sw}, 
to each point of the Sato (universal) Grassmannian
one associates a tau function,
whose second logarithmic derivative yields
a solution to the KdV hierarchy.
\footnote{Precisely,
this yields a solution to the KP hierarchy in general;
it is a solution to the KdV hierarchy 
iff it depends only on the odd-indexed variables.
}
If one takes the point of the Sato Grassmannian
corresponding to the partition $\nu = (g, g-1, \cdots, 1)$,
then the associated tau function is given by
the Schur function $F_{\nu}$ of $\nu$
(cf. \cite[\S 8]{sw}).
By the very definition \eqref{tau-rKdV},
we have an identity
$\tau_g(\vec{a})
=(-1)^{\frac{g(g+1)}{2}}F_{\nu}(a_1, 0, a_2, 0, a_3, \cdots)$, 
where $F_{\nu}$ is considered as a function 
in $t_1, t_2, \cdots$ through \cite[(8.4)]{sw}.
Thus, our function $\tau_g$,
which shares the same second logarithmic derivative
with $F_{\nu}$, 
yields a  rational solution to the KdV hierarchy.
%
%
%
%
%
%
%(3) In the Grassmannian approach to the KdV equation \cite{sw}, solutions are constructed as the second logarithmic
%derivative of the tau function, which is the determinant of an (in general) infinite matrix. When the matrix is
%finite the determinant is a Schur function: with the notation of \cite[\S 8]{sw}, we have an identity
%$\tau_g(\vec{a})=(-1)^{\frac{g(g+1)}{2}}F_{\nu}(a_1, 0, a_2, 0, a_3, \cdots)$, where $F_{\nu}$ is the Schur
%function associated to $\nu = (g, g-1, \cdots, 1)$ considered as 
%a function in $t_1, t_2, \cdots.$ Thus, our function $\tau_g$, which
%provides rational solutions to the genus $g$ Mumford system is a Schur function, yielding a  rational
%solution to the KdV hierarchy.
\end{proof}
%

%\begin{remark}\label{rem:KdV}
%Proposition \ref{prop:KdV} shows that, for every $g \in \Z_{>0}$, the rational function
%$2\rho_g(x_1,x_3,\dots,x_{2g-1})$ is a genus $g$ rational solution to the KdV hierarchy, which coincides, in view
%of Remark \ref{rem:KdV} (2), with the one constructed in \cite{Adler-Moser78} (with all $c_i=0$).
%\end{remark}

\begin{remark}
In \cite[p. 47-48]{sw}, a relation between the Sato Grassmannian and the compactified generalized
Jacobian $\bar{J}_g$ is discussed.  To be more precise, we introduce the Grassmannian $\Gr_g$ of $g$-dimensional
subspaces $W$ of $\C[t]/(t^{2g}) (\cong \C^{2g})$ satisfying $t^2 W \subset W$.  (One can consider $\Gr_g$ as a
subvariety of the usual Grassmannian $\Gr(g, 2g)$ or of the Sato Grassmannian.)  Then $\Gr_g$ admits a 
cell decomposition $\Gr_g=\sqcup_{k=0}^g\Gr_g^{(k)}$ with $\Gr_g^{(k)}\cong\C^k~(k=0,1,\cdots,g)$, and there exists
a bijective morphism $\nu_g: \Gr_g \to \bar{J}_g$ satisfying $\nu_g(\Gr_g^{(k)}) = (\pi_{k, g})_*(J_k)$ for all
${k=0, 1, \cdots, g}$ (cf. Lemma \ref{cpt}). In particular, the cell decomposition of the Grassmannian
corresponds to the stratification of the zero level set $M_g(0)$ of the genus $g$ Mumford system.  Note that
$\nu_g$ is not an isomorphism already in the case $g=1$.  It is conjectured in \cite{sw} that $\nu_g$ gives the
normalization of $\bar{J}_g$.
\end{remark}

%%%%%%%%%%%%%%%%%%%%%%%%%%%

\end{document}